\documentclass[10pt,conference]{IEEEtran}
\IEEEoverridecommandlockouts
\pdfoutput = 1
\usepackage{cite}
\usepackage{amsmath,amssymb,amsfonts}
\usepackage{algorithmic}
\usepackage{graphicx}
\usepackage{textcomp}
\usepackage{xcolor}
\usepackage{multicol}
\def\BibTeX{{\rm B\kern-.05em{\sc i\kern-.025em b}\kern-.08em
    T\kern-.1667em\lower.7ex\hbox{E}\kern-.125emX}}
\usepackage{amsthm}
\newtheorem{theorem}{Theorem}
\newtheorem{prop}{Proposition}
\usepackage[english]{babel}

\def\copyrightpage{\clearpage
\pagestyle{empty}
\begin{minipage}[b]{\textwidth}
\raggedright
\setlength{\parskip}{\baselineskip}
\bigskip
{\Huge IEEE Copyright Notice} \par
\bigskip
{\Large \copyright \, 2020 IEEE.  Personal use of this material is permitted.  Permission from IEEE must be obtained for all other uses, in any current or future media, including reprinting/republishing this material for advertising or promotional purposes, creating new collective works, for resale or redistribution to servers or lists, or reuse of any copyrighted component of this work in other works.} \par
\bigskip\bigskip
{\Large \textbf{Accepted to be Published in: Proceedings of the 2020 IEEE International Symposium on Hardware Oriented Security and Trust (HOST),
May 4-7, 2020, San Jose, CA, USA} }
\end{minipage}
}

\begin{document}

\title{RS-Mask: Random Space Masking as an Integrated Countermeasure against Power and Fault Analysis}

\author{\IEEEauthorblockN{Keyvan Ramezanpour, Paul Ampadu, and William Diehl}
\IEEEauthorblockA{\textit{The Bradley Department of Electrical and Computer Engineering} \\
\textit{Virginia Tech}\\
Blacksburg, VA 24061, USA \\
\{rkeyvan8,ampadu,wdiehl\}@vt.edu}}

\copyrightpage
\maketitle

\begin{abstract}
    While modern masking schemes provide provable security against passive side-channel analysis (SCA), such as power analysis, single faults can be employed to recover the secret key of ciphers even in masked implementations. In this paper, we propose random space masking (RS-Mask) as a countermeasure against both power analysis and statistical fault analysis (SFA) techniques. In the RS-Mask scheme, the distribution of all sensitive variables, faulty and/or correct values is uniform, and it therefore protects the implementations against any SFA technique that exploits the distribution of intermediate variables, including fault sensitivity analysis (FSA), statistical ineffective fault analysis (SIFA) and fault intensity map analysis (FIMA). 
    We implement RS-Mask on AES, and show that a SIFA attack is not able to identify the correct key. We additionally show that an FPGA implementation of AES, protected with RS-Mask, is resistant to power analysis SCA using Welch's t-test.  The area of the RS-Masked AES is about 3.5 times that of an unprotected AES implementation of similar architecture, and about 2 times that of a known FPGA SCA-resistant AES implementation. Finally, we introduce \textit{infective} RS-Mask that provides security against differential techniques, such as differential fault analysis (DFA) and differential fault intensity analysis (DFIA), with a slight increase in overhead.
\end{abstract}


\section{Introduction}

Physical implementations of cryptographic algorithms, even with desirable algorithmic security, can leak information about secret data. In side-channel analysis (SCA), an attacker analyzes the data-dependent signatures or internal states of a cryptographic device to infer secret information. It is therefore critical to protect hardware implementations against SCA to achieve the promised security level of cryptographic algorithms.

In passive SCA, an attacker observes signals leaked from a device processing secret data. 
These techniques include simple power analysis (SPA) and differential power analysis (DPA) \cite{kocher1999differential, mahanta2015power, fabvsivc2016simple, chakraborty2017correlation, luo2018power}. Correlation power analysis (CPA) \cite{brier2004correlation} generalizes DPA by evaluating the correlation of the power samples with a mathematical leakage model, such as Hamming weight (HW) of the processed data \cite{coron2000statistics}, Hamming distance (HD) \cite{brier2004correlation, li2008enhanced} or the number of glitches during the operation of non-linear functions \cite{mangard2005successfully}. Deep learning has also been employed in \cite{timon2019non} to extract higher order statistics of power traces.
In contrast to model-based analysis, mutual information analysis (MIA) in \cite{gierlichs2008mutual} and template attacks \cite{choudary2013efficient}, assume a probability distribution, such as multivariate Gaussian, with data-dependent parameters.

Fault analysis (FA) is a popular and powerful active SCA technique in which a fault injected into the operations of a cipher might result in an error depending on the values of internal secret states \cite{robisson2007differential}. As such, a data-dependent response is induced in the hardware via fault injection.
Differential fault analysis (DFA) inspects the data flow in the internal state of a cipher implementation \cite{liu2017hybrid, li2016impossible, siddhanti2017differential}. The strong adversary model in DFA is difficult to employ on modern nonce-based ciphers.

Statistical fault analysis (SFA) relaxes the assumptions of DFA to achieve more powerful FA attacks. SFA techniques such as fault sensitivity analysis (FSA) \cite{li2010fault}, differential fault intensity analysis (DFIA) \cite{ghalaty2014differential}, ciphertext-only fault analysis (CFA) \cite{li2018ciphertext}, fault intensity map analysis (FIMA) \cite{ramezanpour2019fima} and one version of statistical ineffective fault analysis (SIFA) \cite{dobraunig2018sifa}, exploit the bias induced in the distribution of sensitive variables and/or their differences, via faults leading to timing failure of logic circuits. Biased fault injection, such as laser beams, have also been employed in \cite{dobraunig2016statistical} to induce a biased distribution of variables. Fault attacks such as \cite{moro2013electromagnetic, korak2014effects, yuce2016software} and a second version of SIFA in \cite{dobraunig2018statistical} induce a biased distribution of values by corrupting the instruction registers of processors or internal registers of non-linear functions.

A popular countermeasure against passive SCA that provides provable security is masking, in which secret variables are split into random independent shares. Cipher computations are conducted on separate shares, which are independent of secret data, in a similar way as multi-party computation protocols \cite{prouff2011higher}. 
A widely accepted masking scheme providing security against $d$-th order SCA with $d+1$ shares is the Threshold Implementation (TI) \cite{de2015higher, moradi2011pushing}. The main properties of TI schemes that provide provable security are \textit{non-completeness} (i.e., any combination of $d$ shares of a masked function is independent of at least one share of data), \textit{correctness} (i.e., the final result is correct), and \textit{uniformity} (i.e., output statistics match the input statistics).  Hence, no $d$-th order statistics of the observed signals leak information about the secret.

Most existing countermeasures against DFA employ redundancy to detect an error in computations. Different forms of redundancies include redundant arithmetic \cite{lee2019high, kermani2017reliable, mozaffari2014fault}, time and spatial redundancies \cite{malkin2006comparative, patranabis2016fault}, and information redundancy using coding theory \cite{schneider2016parti, dofe2016comprehensive} or MACs \cite{reparaz2018capa}. However, redundancy-based countermeasures cannot protect against SFA techniques that do not need faulty values for analysis, such as FSA, SIFA and FIMA.

In this paper we present \textit{random space masking} (RS-Mask) as a countermeasure against both passive SCA and fault analysis, and implement this countermeasure in the Advanced Encryption Standard (AES). The linear operations of the cipher are protected in the same way as in a masked scheme. The S-box computations are carried out in a random space of intermediate variables with uniform distribution. The mapping of intermediate variables to the random space is designed such that the output of the S-box is the correct value covered with a Boolean mask. Therefore, any bias induced by fault injection will be randomized with uniform distribution. This countermeasure can be adapted for other ciphers by tailoring proper mapping of the variables to a random space for the specific non-linear operations involved.

We verify the fault resistance of the RS-Masked AES by observing the distribution of sensitive variables under fault injection, and attempting to recover a secret key through computer simulations of SIFA. We further verify its resistance to power analysis SCA through the Test Vector Leakage Assessment (TVLA) methodology (i.e., Welch's t-test) using an open-source test bench on the Artix-7 FPGA. We also benchmark the RS-Masked AES implementation, and quantify the overhead of the new countermeasure.

Our contributions in this work are as follows:
\begin{enumerate}
    \item We introduce an integrated countermeasure, Random Space Masking (RS-Mask) with provable security against a bounded-complexity adversary leveraging power analysis SCA and absolute SFA attacks.
    \item We implement RS-Mask on the popular AES cipher, and demonstrate its resistance to statistical ineffective fault analysis (SIFA) and power analysis SCA.
    \item We derive conditions for provable security against differential fault attacks, such as DFA and DFIA, and introduce \textit{infective} RS-Mask to provide security against these attacks.
\end{enumerate}

The paper is organized as follows. Section \ref{sec:back} reviews existing countermeasures against power and fault analysis. Section \ref{sec:idea} introduces random space masking. The mathematical derivations and the corresponding hardware implementation is discussed in Section \ref{sec:implementation}. The security proof of the RS-Mask scheme agasint statistical fault analysis is given in Section \ref{sec:proof}. Infective RS-Mask against differential fault analysis is introduced in Section \ref{sec:infective} and the results are shown in Section \ref{sec:results}. The paper concludes in Section \ref{sec:conclusion}.

\section{Background and Related Work} \label{sec:back}
\subsection{Combined Countermeasures for Power and Fault Analysis}
A first solution to protect against both power and fault analysis is to employ separate countermeasures which provide protection against different types of attacks. As an example, \cite{patranabis2019lightweight} employs redundant hardware along with shuffling of cipher operations to achieve protection against both fault and power analysis. It also employs fault space transformation (FST) in which the computations of redundant states are carried out in different domains, thus, making it difficult to induce the same error in the redundant states \cite{patranabis2016fault}. 

Shuffling of computations randomizes the timing of internal operations of a cipher, hence, making it difficult for an attacker to align the power traces. Shuffling of datapath is employed in \cite{li2019securing} to protect an AES implementation against localized EM fault attacks. Spatio-temporal randomization on reconfigurable hardware is also used in \cite{wang2016against} to protect different ciphers such as AES against double fault injections. However, shuffling cannot provide perfect security against SCA. Data-dependent features still exist in the power trace even with randomized timing which can be extracted using machine learning algorithms \cite{picek2017side,hettwer2019applications}. Further, considering fault analysis, shuffling randomizes the fault location and/or timing which can be tolerated in most statistical analysis techniques.

In addition to larger overhead, a major drawback of using separate countermeasures against power and fault analysis is the possibility of a negative impact of countermeasures against one type of attack on the other type, if not designed properly. Several works have shown that fault countermeasures employing concurrent error detection (CED) codes facilitate power analysis \cite{regazzoni2007power, regazzoni2008can}. 
The effect of hardware redundancy and parity-based error detection is also studied in \cite{dofe2016comprehensive}. The results show that both types of countermeasures for fault detection increase the speed of key recovery using CPA. 

\subsection{Masked Redundancy}
To alleviate the negative impact of fault detection on power analysis resistance, \cite{schneider2016parti} employs a concurrent error detection scheme in a threshold implementation, called ParTI. A systematic code is employed to generate shares of parity bits at the input and output of a cipher operation. A predictor logic predicts the correct parity bits at the output of the operation from the input parity shares. The calculated and predicted parity bits are compared to detect an error. Rather than concurrent error detection, Countermeasure Against Physical Attacks (CAPA) \cite{reparaz2018capa} and Masks and MACs (M\&M) \cite{de2019m} employ information-theoretic MACs to detect errors as a result of fault injection.

Coarse-grained error detection, at the level of cipher operations, as employed in ParTI, CAPA and M\&M, are not sufficient to protect against most SFA techniques, such as FSA, SIFA and FIMA. A solution is using error correction codes (ECC) to correct all faulty values. In \cite{cryptoeprint:2019:515}, a 3-repetition Hamming code is employed at the gate level to correct any faulty node in the S-box of GIFT-64 cipher. The area of this scheme increases by a factor of $6\times$ in an ASIC implementation. The countermeasure introduced in \cite{cryptoeprint:2019:515} is only against fault analysis. Further, the ECC itself processes sensitive data and can be a target of fault attack.

Fine-grained error detection schemes, at the level of internal nodes of a logic function, similar to fine-grained ECC, can also be employed to protect against SIFA. Such a scheme is employed in \cite{daemenprotecting} with Toffoli-gates that implement non-linear permutations. In this work, 3-, 4- and 5-bit S-boxes that can be represented as non-linear permutations are protected against single-fault attacks with redundant masked Toffoli-gates. It is shown that in a masked Toffoli-gate, if a fault propagates to the input of at most one AND gate, then the event of a correct computation depends only on one share of the input which are independent of sensitive variables. It is argued that to protect against a fault injection that affects $d$ wires of the logic function, $d+1$ redundant masked Toffoli-gates are required, which incurs a large overhead.

\subsection{Integrated Countermeasures}
An alternative to the error detection/correction schemes is randomizing the effect of fault, while providing SCA security, introduced in \cite{bringer2014orthogonal}, which is called orthogonal direct sum masking (ODSM). In this scheme, the sensitive variables of the cipher are encoded with an orthogonal binary linear code. A random value from the complementary domain of the code is also added to the sensitive variables as a random mask. Exploiting the orthogonality property of the code, the sensitive variables can be recovered from the masked version using generating matrix of the code. By deriving the equivalent operations of the cipher in the encoded/masked domain, the correctness of computations is ensured. As the computations are conducted on masked values, the leakage signals are assumed to be independent of the secret. By adding the random mask, the effect of fault is also randomized, thus, resistance against statistical analysis is achieved.

Although ODSM promised a generic countermeasure against SCA and fault analysis, studies of \cite{barbu2016analysis} demonstrated that ODSM fails to provide the assumed security. The major vulnerability of ODSM arises from the fact that the random mask added to sensitive variables is not perfectly uniform as opposed to Boolean masking, since the mask is chosen from the orthogonal space to the encoded sensitive variables. As a result, the distribution of masked variables cannot be uniform. This property is exploited in \cite{barbu2016analysis} to deploy a first order DPA attack to recover the secret key of AES implemented with ODSM scheme. We will also prove that the non-uniform distribution of a mask makes the scheme vulnerable to fault analysis.

\section{Random Space Masking}  \label{sec:idea}
\subsection{Primary Idea}
The basic idea of RS-Mask is to map the intermediate variables of a cipher to a random space. After operations in the random space, the output is transformed back to the original space of the cipher. As a result, if the random mapping is uniform, the intermediate variables on which the cipher operations are conducted become independent of the sensitive data, hence, the implementation is robust against SCA. When a fault is injected into intermediate variables in a random space, any effect that the fault has on the data depends on random values which are independent from the sensitive data. The main challenge in implementing such a scheme is finding the proper random mapping, such that after performing the non-linear operations of the cipher in the random space, the correct values of the variables can be recovered.

\begin{figure}[t!]
	\vspace{-0.2cm}
	\centering
	\includegraphics[width=0.4\textwidth]{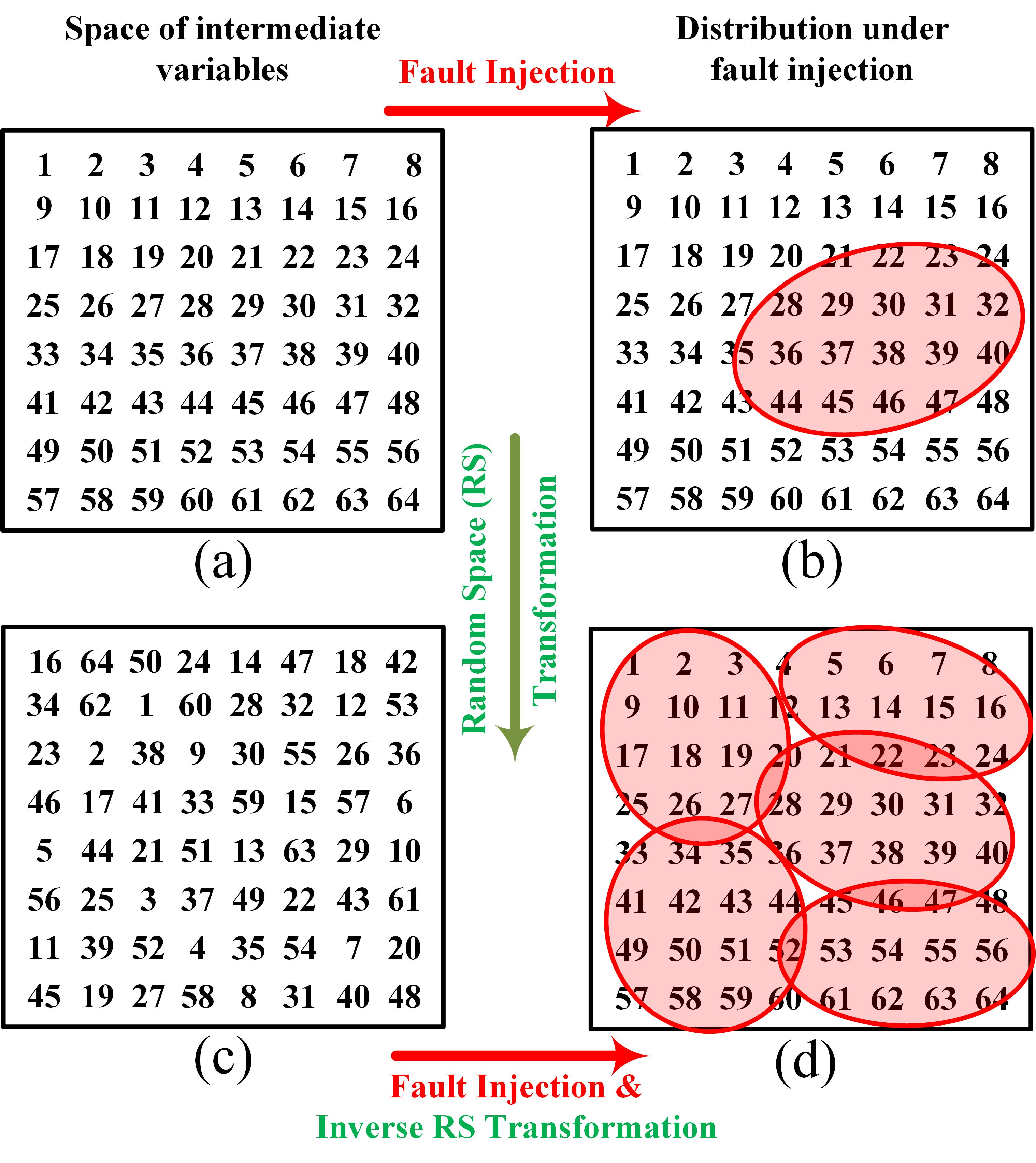}
	\vspace{-0.3cm}
	\caption{Conceptual representation of random space (RS) masking that randomizes the effect of fault; shaded area is the distribution of values under fault injection.}
	\label{fig:rsmap}
	\vspace{-0.4cm}
\end{figure}
The concept of RS-Mask is depicted in Fig. \ref{fig:rsmap}. The effect of a biased fault is shown in part (b) of the figure. Fault injection maps the space of internal variables of a logic function to a particular set of values, which is depicted as the shaded area in the figure. If the input to the function results in the values of internal variables in the shaded area, the fault injection will be ineffective; the result of computations is still correct under fault injection. Depending on the fault type and location, only particular values might result in this set of internal values, hence, the distribution of correct values under fault injection will be non-uniform, or biased. Similarly, distribution of faulty values might also be biased.

In RS-Mask, we map the input of a logic function to any possible value with equal probability, as shown in part (c) of Fig. \ref{fig:rsmap}. As a result, during the operation of the logic function under attack, any input can result in all possible values of internal variables with equal probability. Hence, the distribution of faulty values and correct values under ineffective faults will be uniform.

\subsection{Masking as a Random Mapping}
Masking can be considered as a random mapping for sensitive variables. 
In a masking scheme, a sensitive variable $x$ is split into $d+1$ shares $x_i, i=0,1,\cdots,d$ such that the sum of all $x_i$'s is equal to the correct value $x$ and any $d$ combination of the shares are uniformly distributed. A function $y=f(x)$ is split into $d+1$ function shares $y_i=f_i(x_0,x_1,\cdots,x_d)$ such that the sum of $y_i$'s is equal to the correct result $y$ with any $d$ combination of $y_i$'s uniformly distributed. While linear functions process input shares separately, a non-linear function share $f_i$ is possibly a function of all input shares.

Since, for linear functions, each share of the function depends only on one share of the sensitive variables, if an attacker injects biased faults into any combination of $d$ shares, there is still one correct share with a uniform distribution. Hence, the sum of all shares is still uniformly distributed under fault injection. Therefore, masking can be considered as a sound countermeasure against statistical fault analysis techniques with $d$ faults. 

\begin{figure*}[htbp]
	\centering
	\begin{multicols}{2}
		\includegraphics[width=0.9\textwidth]{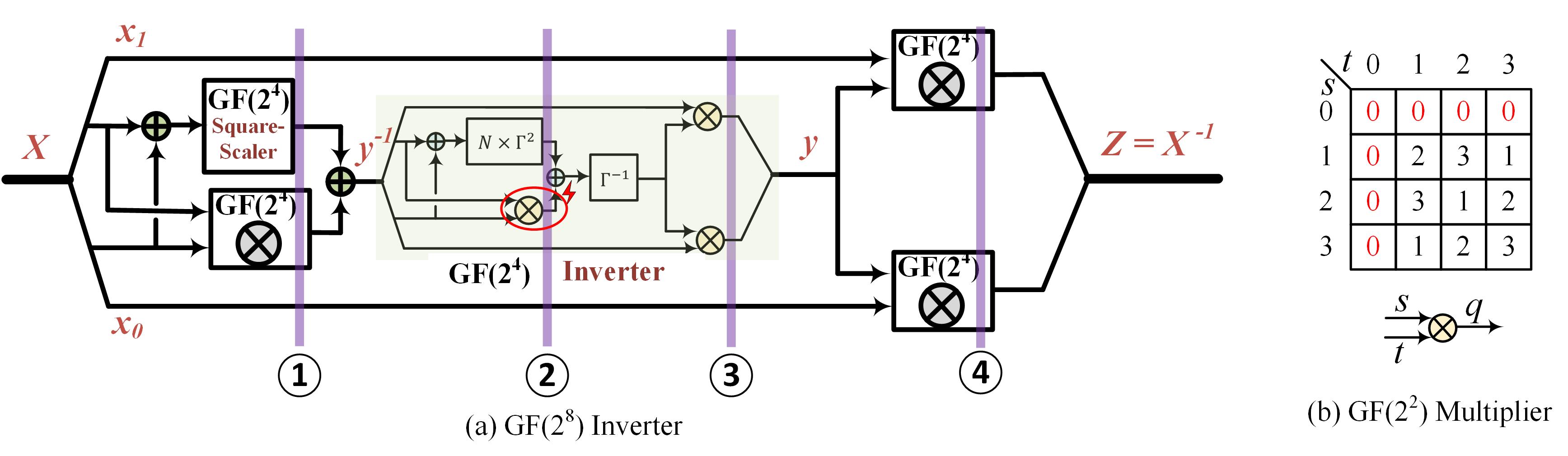}
	\end{multicols}
	\vspace{-0.8cm}
	\caption{Canright composite field implementation of $GF(2^8)$ inverter in AES S-box, (a) inverter block diagram, (b) $GF(2^2)$ multiplier, $q=s\times t$.}
	\label{fig:inverter}
\end{figure*}

However, this argument does not hold for non-linear functions. First, each share of a non-linear function depends on all shares of the input in most masked implementations. Hence, a faulty value of a single input share can propagate to all output shares, possibly making all biased. Second, to induce bias into most non-linear functions, such as S-boxes, it is not necessary to employ a biased fault mechanism; any type of fault can result in a biased distribution of the function output.

The reason a random fault can induce a biased distribution in most non-linear functions is the non-bijective constituent components of a non-linear function. In the case of AES S-box, the Canright composite field computation for the $GF(2^8)$ inverter, as shown in Fig. \ref{fig:inverter}, is popular due to its compact implementation \cite{canright}. The basic non-linear building blocks of this implementation are $GF(2^2)$ multipliers. The truth table of the multiplier is shown in Fig. \ref{fig:inverter} (b), which represents 2-bit multiplication $q=s\times t$ in the normal basis of Canright scheme. The multiplier is a non-bijective function of an input $t$, if $s=0$. This is the origin of the bias at the output of the multiplier even for uniform faults.

The \textit{uniformity} of masking schemes, i.e., uniform distribution of the shares for all intermediate variables, is a fundamental property to provide security against SCA. However, the uniformity property alone does not help in protecting against biased-fault analysis.  One might speculate that both the \textit{non-completeness} and \textit{uniformity} properties can inhibit a bias as a result of single faults, because there exists at least one share of the function which is independent of a given input share, and has a \textit{uniformly} distributed output. As we will show, this argument is true only if the function shares are implemented as atomic operations, and uniformity is not achieved via remasking.

\begin{figure}[t!]
	\centering
	\includegraphics[width=0.4\textwidth]{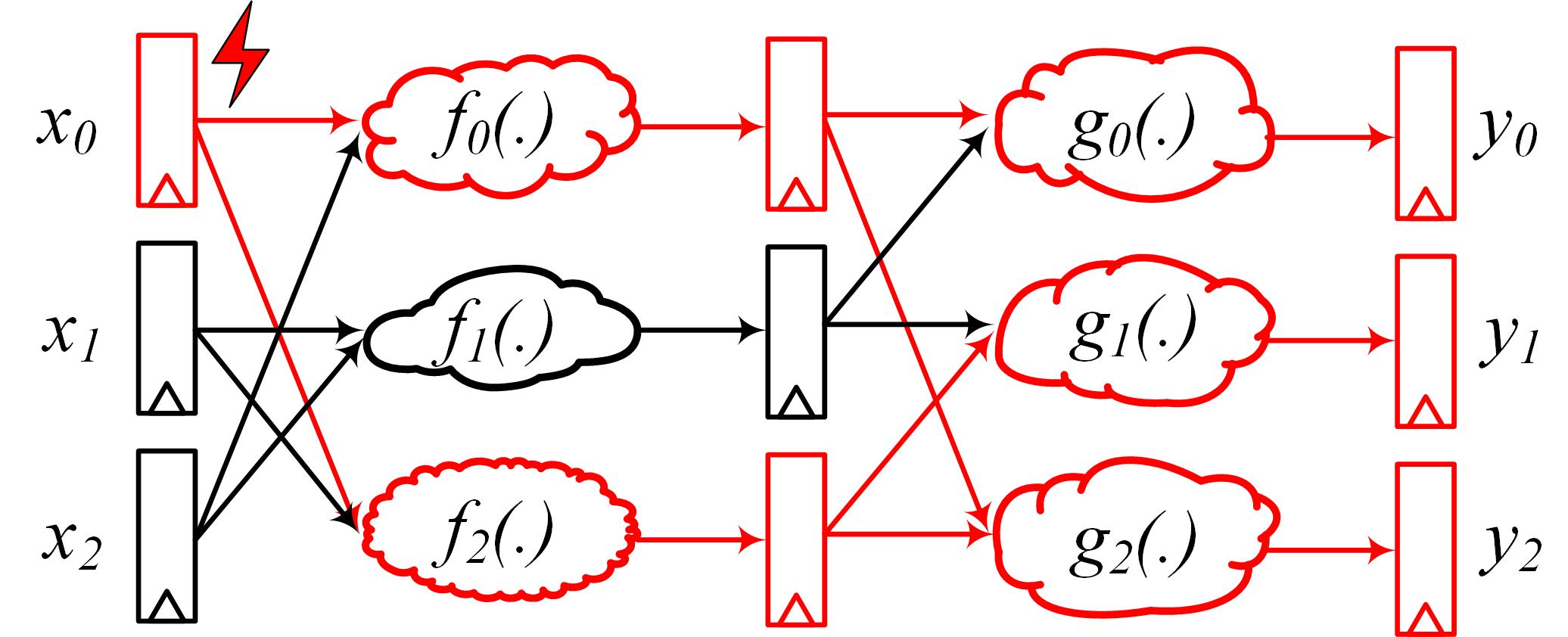}
	\vspace{-0.3cm}
	\caption{Propagation of a fault in a single share of internal variables to all shares of output.}
	\label{fig:prop}
\end{figure}

As most implementations of non-linear functions of a cipher consist of a series of non-linear operations, the final output shares depend on all shares of variables at the early stages of the function, even with the non-completeness property. In Fig. \ref{fig:prop}, a non-linear function is implemented as a series of two components $f(\cdot)$ and $g(\cdot)$, each split into three shares. The sharing also satisfies the non-completeness property; each function share is independent of one share of its input. We observe that a fault injected at a single share of the input propagates to all shares of the output. Due to the non-linearity of the functions, all output shares might be biased.

If a single fault is injected at the input of the function component $g(\cdot)$ in Fig. \ref{fig:prop}, there is at least one share of the output which is fault-free and uniform. One might conclude that the single fault, at this location, cannot induce a bias into the output of $g(\cdot)$. This is true only if the output shares are uniform by the design of function shares $g_i(\cdot), i=0,1,2$. 
Remasking is a popular technique to achieve uniformity, in which uniform refreshing masks are added to the output shares of non-linear functions. However, remasking is not sufficient to protect against biased fault analysis, since the refreshing masks are cancelled out after combining the shares. Further, adding linear combinations of input shares to the output of non-linear shares, as in \cite{nikova2006threshold}, is equivalent to remasking, since these additional terms also cancel out after combining the shares.

In a masking scheme in which the function shares result in uniform output shares by design, and not using remasking, the non-completeness property might help inhibit a bias only if \textit{single} faults are injected at final stages of calculations; faults at early stages can quickly propagate to all output shares, as shown in Fig. \ref{fig:prop}, leading to biased distribution of all shares. Random space masking is a solution to achieve uniform distribution at the output of a bijective non-linear function under any type of fault attack, via pre-randomization of the input to the function. 

\subsection{Random Space (RS) Mapping}
In the proposed RS-Mask scheme, the linear functions of a cipher are protected according to conventional masking schemes. One of the shares, called RS share, is always independent of the sensitive data shares; neither the secret key nor any key-dependent intermediate variables are ever combined with the RS share during cipher operations. The RS share can be further split into multiple shares to increase the complexity of a fault attack. The sub-shares of the RS share can be combined at any stage of the cipher operations to reduce the overhead of computations.

The non-linear function of the cipher, i.e., the S-box in the case of AES, is protected against biased fault analysis using RS mapping. The general block diagram of random mapping is shown in Fig. \ref{fig:rsnonlin}. At the input of the non-linear function $g(\cdot)$, the sensitive variable $X$ is transformed to $X^{'}$ with a random mapping $\mathcal{E}_1$. The correct output $Z$ is recovered from the output $Z^{'}$ by applying the dual mapping $\mathcal{E}_2$.

The necessary and sufficient conditions for the RS mappings $\mathcal{E}_1$ and $\mathcal{E}_2$ to protect the bijective function $g(\cdot)$ against biased fault analysis are as follows:
\begin{itemize}
    \item \textit{Correctness}: We must have $\mathcal{E}_2\circ g(\mathcal{E}_1\circ X) = g(X)$.
    
    \item \textit{Uniformity}: The mapping $\mathcal{E}_1$ must be uniform; i.e, for any value $X\in \mathcal{F}^n$ the result of $\mathcal{E}_1\circ X$ must be any value in $\mathcal{F}^n$ with equal probability.
    
    \item \textit{Robustness}: The mapping $\mathcal{E}_2$ itself must be robust against biased fault analysis; i.e., the distribution of values at the output of $\mathcal{E}_2$ must be uniform under fault injection without remasking.
\end{itemize}

\begin{figure}[t!]
	\centering
	\includegraphics[width=0.45\textwidth]{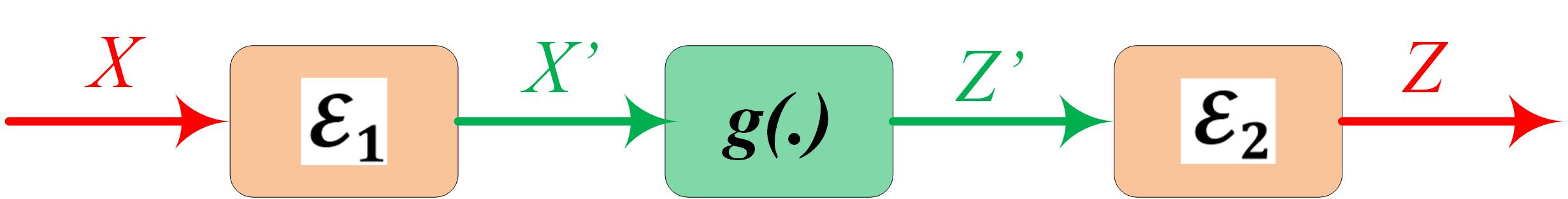}
	\caption{General block diagram of random mapping to protect a bijective non-linear function $g(\cdot)$ against biased fault analysis.}
	\label{fig:rsnonlin}
\end{figure}

The main challenge in designing an RS-Mask for a given cipher is finding the proper mappings $\mathcal{E}_1$ and $\mathcal{E}_2$ that satisfy the above conditions. In this work, we define the random mapping $\mathcal{E}_1$ such that
\begin{equation}
    g(\mathcal{E}_1\circ X)=g(X)\oplus R
\end{equation}
in which $R$ is a uniformly distributed mask. The output is thus comprised of two shares, i.e., the calculated value $Z^{'}$ and a random mask $R$. Such a scheme satisfies the above conditions with the following advantages:
\begin{itemize}
    \item The result of the function $g(\cdot)$ on the transformed input $X^{'}$ is a linear combination of a Boolean mask and the correct output $Z$ which is compatible with traditional masking schemes. Hence, there is no need to convert the output to a masked version required for subsequent steps of the cipher operations.
    
    \item No non-linear operation, such as $\mathcal{E}_2$, is required to recover the correct value from the random space. As a result, the robustness property is satisfied. The uniformity of the masked output is inherent in designing the proper $\mathcal{E}_1$.
\end{itemize}

We notice that, in the proposed RS mapping scheme, fault injection at the input mapping $\mathcal{E}_1$ cannot induce a bias into the correct output $Z$. The calculated output $Z^{'}$ should normally be uniform. If a fault injection, either in $\mathcal{E}_1$ and/or $g(\cdot)$, induces a bias into the distribution of $Z^{'}$, the correct value $Z=Z^{'}\oplus R$ would still be uniform. 

\section{Implementation of RS-Mask}  \label{sec:implementation}
\subsection{Derivation of RS Mapping}
In order to protect the $GF(2^8)$ inverter in the AES S-box against biased fault analysis, we need a random mapping on the input $X$ such that the output is $Z^{'}=X^{-1}\oplus R$ with a uniformly distributed $R$. 
By representing the input and output of the $GF(2^8)$ inverter with their most and least significant nibbles as $X=x_1||x_0$ and $Z=z_1||z_0$, we have $z_1 = x_0\otimes y$ and $z_0 = x_1\otimes y$, in which $y$ is the output of the $GF(2^4)$ inverter as shown in Fig. \ref{fig:inverter}, and $\otimes$ is $GF(2^4)$ multiplier. We find the desired random mapping in two cases of zero and nonzero input values. 

\subsubsection{Zero Input}
We note that the inverse of $X=0$ in the field $GF(2^m)$ is $0$ while the inverse of all nonzero values of $X$ is nonzero. According to the block diagram of Fig. \ref{fig:inverter}, if $y=0$, then the output of the inverter is $0$. Hence, $y=0$ only if $X=0$. Assuming $y\ne 0$, which corresponds to a nonzero input to the $GF(2^8)$ inverter, we can write
\begin{align} \label{eq:map1}
\begin{split}
    z_1^{'} = z_1 \oplus r_1 = [x_0 \oplus (y^{-1}\otimes r_1)]\otimes y \\
    z_0^{'} = z_0 \oplus r_0 = [x_1 \oplus (y^{-1}\otimes r_0)]\otimes y \\
\end{split}
\end{align}
in which $R=r_1||r_0$ is the random mask. The terms in the brackets can be considered as the required mapping to achieve the desired output. The value of $y^{-1}$ is also available at the input of the $GF(2^4)$ inverter.

The mapping in (\ref{eq:map1}) is not sufficient to protect against faults injected at the earlier operations of the inverter. Assume $y^{-1}$ takes a faulty value, denoted by $(y^{-1})^{*}$, then the output of $GF(2^4)$ inverter also takes a faulty value $y^{*}$. Since the operations of $GF(2^4)$ inverter is fault-free, we have $(y^{-1})^{*}=(y^{*})^{-1}$. Using this relation, the faulty value of the output in (\ref{eq:map1}) reads
\begin{equation}
    (z_1^{'})^{*} = [x_0 \oplus ((y^{*})^{-1}\otimes r_1)]\otimes y^{*} = x_0 \otimes y^{*} \oplus r_1
\end{equation}
A similar equation also holds for the least significant nibble of the output. At the end of cipher operations, the mask $r_1$ will be combined with the calculated value $(z_1^{'})^{*}$. Hence, the recovered value will be $x_0 \otimes y^{*}$ which is clearly biased.

To prevent the leakage of a faulty value to the output of the non-masked output $Z$ in (\ref{eq:map1}), we must use a separate datapath for calculating $y^{-1}$ used in the random mapping. According to the block diagram of Fig. \ref{fig:inverter}, we can write 
\begin{equation} \label{eq:yinv}
    y^{-1} = \nu(x_0\oplus x_1)^2 \oplus (x_0\otimes x_1)
\end{equation}
in which $\nu\in GF(2^4)$ is a scaling constant in the Canright scheme. By substituting (\ref{eq:yinv}) into (\ref{eq:map1}), we can derive the desired random mapping. However, to derive a single mapping for all input values we will need an auxiliary variable as discussed later.

\subsubsection{Zero Input}
When $X=0$, then $y=y^{-1}=0$ and the equations in (\ref{eq:map1}) do not hold. Hence, the random mapping in (\ref{eq:map1}) is not correct for zero input. In this case, we use the linearity property of the $GF(2^8)$ for $X=0$, i.e.,
\begin{equation}
    (0 \oplus R)^{-1} = 0^{-1} \oplus R^{-1}
\end{equation}
Using this property, if we add $R^{-1}$ to the input, when it is zero, the output of the inverter will be the masked correct value $Z^{'}=0\oplus R=R$. However, in this case, the mapping in (\ref{eq:map1}) must not be applied during the calculations of the inverter. Similarly, when the input is nonzero, $R^{-1}$ must not be added to the input. The above derivation is explained separately for zero and nonzero inputs for clarity. However, the final mapping processes all input values within the same datapath as discussed below.

\subsubsection{All Input Values}
We use an auxiliary variable $f$ to achieve a single random mapping for all values of the input $X$ defined as 
\begin{equation}
    f = 
    \begin{cases}
    15, & X\ne 0 \\
    0, & X=0
    \end{cases}
    \label{eq:condition}
\end{equation}
Note that the value $15$ is the multiplicative unity in the field $GF(2^4)$. Now, with a random mask $R$, we add the value $\Bar{f}\otimes R$ to the input in which $\Bar{f}$ is the bitwise inverse of $f$. In this case, the value of $R^{-1}$ is added to the input only if $X=0$. Similarly, by combining (\ref{eq:map1}) and (\ref{eq:yinv}) with $f$, we get the final random mapping as
\begin{align} \label{eq:fmap}
\begin{split}
    z^{'}_1 = \big(x_0 \oplus [(x_0\oplus x_1)^2\otimes \nu r_0] \oplus [(x_0\otimes r_0)\otimes x_1]\otimes f\big)\otimes y \\
    z^{'}_0 = \big(x_1 \oplus [(x_0\oplus x_1)^2\otimes\nu r_1] \oplus [(x_1\otimes r_1)\otimes x_0]\otimes f\big)\otimes y
\end{split}
\end{align}

\subsection{Hardware Implementation}
\begin{figure*}
	\centering
	\begin{multicols}{2}
		\includegraphics[width=1\textwidth]{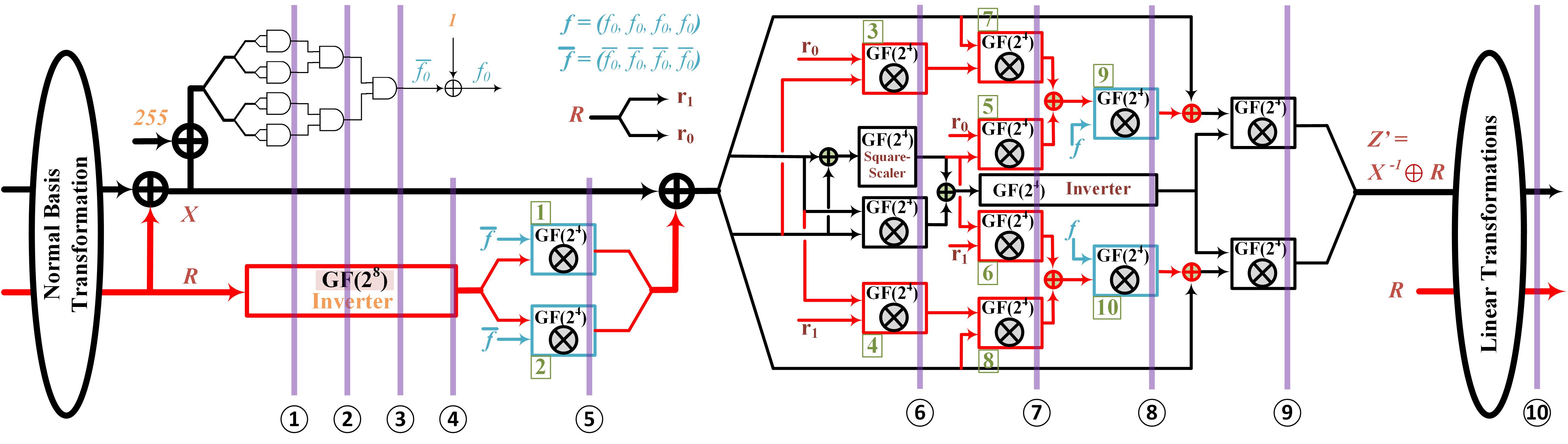}
	\end{multicols}
	\vspace{-0.8cm}
	\caption{RS-Mask Implementation of AES S-box, according to Canright composite field implementation, with a uniformly distributed mask $R$ for random mapping; Red connections represent additional computations required for RS mapping; Numbered vertical lines represent pipeline stages.}
	\label{fig:sbox}
\end{figure*}
The overall block diagram of the AES S-box with composite field implementation in the RS-Mask scheme is shown in Fig. \ref{fig:sbox}. The $GF(2^4)$ multipliers are implemented as a parallel of three $GF(2^2)$ multipliers as shown in \cite{schneider2016parti}. All AND-gates and multipliers in Fig. \ref{fig:sbox} are masked according to TI with three shares. To protect against SCA in the presence of glitches, the output of all non-linear components are registered. The output \textit{linear transformations} include the combined affine transformation and normal basis conversion as well as the \textit{MixColumn} operation of AES. After four bytes of state reach register level 10, the result of one column of the state is ready by applying \textit{MixColumn} and \textit{AddRoundKey}.

In our implementation, the entire state of AES is split into three shares. One of the shares is the RS share $R$. The round keys are also split into two shares with the RS share set to $0$ to save extra randomness and registers. Since fault analysis techniques induce a bias into the distribution of state bytes, it is not necessary to mask the round keys with an RS share. As observed in Fig. \ref{fig:sbox}, the RS share is combined with one share of the data at the input of the S-box, after normal basis conversion. The output of the S-box consists of the processed data and the RS share. At the end of a round, the two shares of the round keys are added to the two shares of data excluding the RS share.

The structure of the $GF(2^8)$ inverter in Fig. \ref{fig:sbox} that processes the random mask $R$ is the same as in Fig. \ref{fig:inverter}. This inverter is implemented with only one share as $R$ is independent of any sensitive variable. In order to prevent an attacker from inferring the value of $R$ by observing the number of glitches, the output of non-linear components in this inverter are also registered, as shown in Fig. \ref{fig:inverter}.

There are two types of $GF(2^4)$ multipliers in the RS-Mask scheme of Fig. \ref{fig:sbox}. In one type, i.e., multipliers 1 to 6, an intermediate variable, with three shares, is multiplied by a random mask which has only one share. In other multipliers, both inputs have three shares. For masking multipliers with a single-share input, the overhead increases linearly with the number of shares, since $x\otimes r = (\Sigma_i x_i)\otimes r = \Sigma_i(x_i\otimes r)$.

The random mapping of the proposed scheme, in Fig. \ref{fig:sbox}, requires four extra multipliers with three shares at both inputs. The overhead of masking six multipliers with a single-share input is equivalent to two fully masked multipliers with three shares. The overall overhead of the random mapping is thus equivalent to six fully masked $GF(2^4)$ multipliers. Including the extra single-share $GF(2^8)$ inverter for the RS mask and the logic for calculating the auxiliary variable $f$, the overall overhead of the RS-Mask S-box is almost equivalent to $3\times$ the area of a TI S-box with the same number of shares.

According to Fig. \ref{fig:sbox}, nine register levels are required, after every stage of non-linear operations, to achieve security against SCA in the presence of glitches. The output shares of every non-linear operation are remasked before being stored in the registers. An extra register level is required at the output of S-box calculations to perform \textit{MixColumn} operation on the bytes of a column. In order to maintain throughput, the AES algorithm is implemented in a 10-stage pipelined design. The same instance of the S-box in Fig. \ref{fig:sbox} is reused for calculating the state bytes and the round keys.

The entire AES encryption algorithm with three shares of the state according to the RS-Mask scheme is implemented on Artix-7 FPGA. A comparison of an unprotected AES, a TI implementation secure against only SCA, and the proposed RS-Mask AES with three shares, secure against both SCA and statistical fault analysis is shown in Table~\ref{tab1}. The relative cost of TI and RS-Mask implementations versus unprotected AES is also given in Table \ref{tab2}. Frequency and area are optimized in Xilinx Vivado using the Minerva automated optimization tool \cite{farahmand2017minerva}. The TI implementation of AES in these tables, available at \cite{aesti}, has two shares for the linear operations and three shares for S-box calculations.

\begin{table}[t!]
\centering
\caption{Comparison of unprotected, TI and RS-Mask AES.}\label{tab1}
\begin{tabular}{|c|c|c|c|}
\hline
\quad Result of Optimization \quad & \quad AES \quad & TI AES & RS-Mask AES \\
\hline\hline
Max frequency [MHz] & 242 & 195 & 218 \\
\hline
LUT & 508 & 1373 & 2273 \\
\hline
Slices & 264 & 583 & 827 \\
\hline
Throughput [Mbps] & 151.1 & 121.8 & 116.8 \\
\hline
Power [mW] @ 20MHz & 31.8 & -- & 44.3 \\
\hline
Energy [nJ/bit] & 2.55 & -- & 4.14 \\
\hline
\end{tabular}
\end{table}

\begin{table}[t!]
\centering
\caption{Overhead of TI and RS-Mask versus unprotected AES.}\label{tab2}
\begin{tabular}{|c|c|c|}
\hline
Cost & \quad\quad TI AES \quad\quad & RS-Mask AES \\
\hline\hline
Decrease in frequency [\%] & 19.4 & 9.9 \\
\hline
Decrease in TPA [\%] & 70.2 & 82.7 \\
\hline
Increase in Area [\%] & 170.3 & 347.4 \\
\hline
Increase in Total Power [\%] & -- & 39.3 \\
\hline
Increase in Dynamic Power [\%] & -- & 375 \\
\hline
\end{tabular}
\end{table}

As observed in Tables \ref{tab1} and \ref{tab2}, the major overhead of RS-Mask is the increase in area and \textit{dynamic} power. The overall area of the AES implementation with RS-Mask is around $3.5\times$ the area of an unprotected AES and $2\times$ the TI AES. The RS-Mask AES consumes 44.3 mW of total power (31 mW static power) measured at 20 MHz on CW305 Artix-7 target board compared with 31.8 mW (29 mW static) power consumption of unprotected AES. With the 10-stage pipeline design of RS-Mask, the throughput is degraded by around 22\% which is mainly due to the decrease in the operation frequency, but also due to the number of clock cycles of RS-Mask increasing to 239 per block, vice 205 per block in the AES unprotected and AES TI. The decrease in throughput per area (TPA) of the RS-Mask scheme is $82.7\%$ which is mainly due to the large area cost, while, the TPA degradation of TI implementation is $70.2\%$, close to that of RS-Mask. 

\section{Security Proof}  \label{sec:proof}
The security of the proposed RS-Mask scheme against passive SCA is based on the security promises of TI masking schemes, as all building blocks of the RS mapping are implemented according to TI. We make no further claims on the security of RS-Mask against power analysis. In this section, we prove the security of the proposed RS-Mask scheme against \textit{absolute} statistical fault analysis. Further, we derive conditions for security of a countermeasure against \textit{differential} fault analysis.

We categorize all statistical fault analysis techniques into two classes of \textit{absolute} and \textit{differential} analysis. We define \textit{absolute} fault analysis as any technique that employs the distribution of an intermediate variable of a cipher under fault injection as the observable. The \textit{differential} fault analysis is defined as any technique that observes the distribution of the difference induced in an intermediate variable as a result of fault injection. The absolute and differential fault analysis techniques are analogous to the classical linear and differential cryptanalysis. In the classical cryptanalysis, the bias is due to the imperfect design of the non-linear operations in a cipher, while, in fault analysis it is the result of fault injection.

The following theorem proves the security of RS-Mask against absolute fault analysis.
\begin{theorem} \label{thm1}
Consider an arbitrary random Boolean variable $X$ and an independent Boolean mask $R$. The mutual information between the masked variable $Z=X\oplus R$ and $X$ is zero if and only if $R$ is uniformly distributed.
\end{theorem}
\begin{proof}
First, assume the random mask $R$ is uniformly distributed. Considering $m$-bit Boolean variables, we can derive the probability mass function (pmf) of $Z$ as a circular convolution of the pmf of $X$ and $R$, i.e.,
\begin{equation}
    p_Z(z) = \sum_{r} p_R(r)\cdot p_X(z\oplus r) = \frac{1}{2^m}\sum_{r} p_X(z\oplus r)
\end{equation}
in which the second equality is due to the uniform distribution of $R$. Since $\sum_x p_X(x)=1$, we have $p_Z(z)=1/2^{m}$. Hence, $Z$ is also uniformly distributed. Now, the mutual information between $Z$ and $X$ can be written as
\begin{equation} \label{eq:mutual}
    I(Z;X) = H(Z) - H(Z|X) = m - m = 0
\end{equation}
where $H(\cdot)$ is the Shannon entropy. Since $Z$ is uniformly distributed, $H(Z)$ takes the maximum value which is equal to $m$ for an $m$-bit variable. Further, given $X$, the distribution of $Z$ is equal to the distribution of $R$ which is also uniform; hence, $H(Z|X)=m$.

To prove the converse of the theorem, we consider a non-uniform variable $R$ and we show that it is not possible that $I(Z;X)=0$. Assume $I(Z;X)=0$. From the definition of mutual information, $I(Z;X)=0$ implies that $H(Z)=H(Z|X)$. Additionally, we know that
\begin{align} \label{eq:prop1}
\begin{split}
    &H(Z|X) = -\sum_x P(X=x) \cdot \big[ \\ 
    &\sum_z P(R=z\oplus x|X=x)\cdot \log\big(P(R=z\oplus x|X=x)\big) \big] 
\end{split}
\end{align}
which is equal to $H(R|X)$ by definition. Since $R$ and $X$ are independent, $H(R|X)=H(R)$. Hence, we conclude that 
\begin{equation} \label{eq:equalentropy}
    H(Z) = H(Z|X)=H(R|X)=H(R)
\end{equation}
On the other hand, from the property that \textit{conditioning reduces entropy}, we have
\begin{align} \label{eq:prop2}
\begin{split}
    H(Z) \ge H(Z|X)=H(R|X)=H(R) \\
    H(Z) \ge H(Z|R)=H(X|R)=H(X) \\
\end{split}
\end{align}
The second equation can be proved the same way as in (\ref{eq:prop1}). From (\ref{eq:prop2}), we conclude that $H(Z)\ge \max\{H(X),H(R) \}$. Since $X$ is an arbitrary random variable, we can choose $H(X)>H(R)$. Hence, $H(Z)>H(R)$. But from (\ref{eq:equalentropy}), we must have $H(Z)=H(R)$ which is a contradiction.
\end{proof}

According to Theorem \ref{thm1}, since the calculated output $Z^{'}=X^{-1}\oplus R$ in Fig. \ref{fig:sbox} is the sensitive variable $X^{-1}$ masked with a uniformly distributed variable $R$, the output is independent of the sensitive variable. To achieve the claimed security, the distribution of the RS mask must be uniform. As a result, by observing the distribution of the calculations at the output of the S-box, no information can be retrieved about the sensitive variables. Similarly, for linear operations of the cipher, masked with $d+1$ shares, observation of at most $d$ variables leaks no information about the sensitive data as there is at least one share that works as a uniform mask.
However, a differential fault analysis technique can still retrieve information about the secret data. The following theorems demonstrate this claim.

\begin{prop} \label{prop1}
Consider a bijective non-linear function $X=S(Y)$ such that for a random variable $Y$ and an independent constant $K$, $X=S(Y\oplus K)$ is uniformly distributed for any nonzero $K$. Consider non-equal random variables $Y_1$ and $Y_2$ such that the distribution of $S(Y_1)\oplus S(Y_2)$ is non-uniform. The distribution of $S(Y_1\oplus K)\oplus S(Y_2\oplus K)$ is uniform if and only if $K\ne0$.
\end{prop}
\begin{proof}
Assume $K\ne 0$. From the definition of the function $S(\cdot)$, we know that $X_1=S(Y_1\oplus K)$ and $X_2=S(Y_2\oplus K)$ are uniformly distributed. From Theorem \ref{thm1}, we can conclude that $X_1\oplus X_2$ is independent of both $X_1$ and $X_2$, thus, is uniformly distributed.

Proof of the converse is trivial. When $K=0$, the difference $S(Y_1)\oplus S(Y_2)$ is non-uniformly distributed as stated in the proposition.
\end{proof}

The above proposition is a statement of the fundamental property of well-designed ciphers. For the example of AES, we can represent the relationship between one byte of the output ciphertext, i.e., $C_i$, and the corresponding byte at the output of round 9, $X_i$, as $X_i=S(C_i\oplus K_{i})$, in which, $K_{i}$ is a constant related to the round key 10 and $S(\cdot)$ is the inverse of the S-box. Although the linear operations of the cipher are discarded, this relation provides a sound model for the non-linear behavior of the cipher. If an incorrect value of the round key 10 is used to calculate the intermediate variables at round 9 from the output ciphertext, the distribution of calculated values will always be uniform.

\begin{theorem} \label{thm2}
Consider arbitrary random variables $Y_1$ and $Y_2$ with a non-uniform difference $S(Y_1)\oplus S(Y_2)$, in which $S(\cdot)$ is a bijective non-linear function with the properties stated in Proposition \ref{prop1}. The mutual information $I(X_1;X_2)$, with $X_1=S(Y_1\oplus K)$ and $X_2=S(Y_2\oplus K)$, is nonzero if and only if $K=0$.
\end{theorem}
\begin{proof}
When $K=0$, it is known that the difference $\Delta =X_1\oplus X_2$ is non-uniformly distributed as stated in the theorem. We have $X_2=X_1\oplus\Delta$. According to Theorem \ref{thm1}, since $\Delta$ is non-uniformly distributed, the mutual information $I(X_1;X_2)$ is nonzero.

To prove the converse, assume $K\ne 0$. According to Proposition \ref{prop1}, the distribution of the difference $\Delta=X_1\oplus X_2$ is uniformly distributed. By representing $X_2=X_1\oplus\Delta$ and from Theorem \ref{thm1}, we conclude that $I(X_1;X_2)=0$.
\end{proof}

The above theorem is a statement of differential fault analysis. Assume $X_1$ and $X_2$ represent the correct and faulty values of an intermediate variable with a biased difference. An attacker has access to the correct and faulty outputs of the cipher, i.e., $Y_1$ and $Y_2$, respectively. To calculate the intermediate variables, the secret key $K$ is required as $X_i=S(Y_i\oplus K), i=1,2$. From Theorem \ref{thm2}, we know that the mutual information between the correct and faulty values is nonzero only for the correct key. If we consider the mutual information as a test statistic for ranking key candidates, the correct key exhibit the highest rank; thus, it can be identified.

\section{Infective RS-Mask}  \label{sec:infective}
While the RS-Mask scheme of Fig. \ref{fig:sbox} randomizes intermediate variables, there is no guarantee that the distribution of differences induced by the fault is also uniform. Hence, differential fault analysis techniques are still able to attack this scheme. Among all SFA techniques, the statistical DFA in \cite{lashermes2012dfa} and DFIA are differential analysis. They both observe the bias in the error distribution. All other SFA techniques, including SIFA and FIMA, are absolute techniques which observe the distribution of the intermediate variables. In this section, we introduce \textit{infective} RS-Mask as an extension, that additionally provides security against differential FA.

The intrinsic redundancy of RS-Mask in Fig. \ref{fig:sbox} can be exploited to detect an error with a small overhead. The output of the RS-Mask S-box is $Z^{'}=Z\oplus R$. We can also calculate $Z$ by using two additional $GF(2^4)$ multipliers as shown in Fig. \ref{fig:redundant}. The difference between $Z'$ and $Z$ must be equal to $R$ if there is no error. The values $x_0$, $x_1$ and $y$ are available from the calculations of Fig. \ref{fig:sbox}. Having a random mask $R_1$ with uniform distribution, we can add $E\times R_1$ to the output of the S-box. If the error is nonzero, the distribution of $E\times R_1$ is uniform; otherwise, $E\times R_1=0$. Hence, if an error occurs, the values of the state bytes are randomized, or infected. With uniform distribution of errors, no differential SFA technique can be successful. The multiplication $E\times R_1$ can also be optimally implemented to have an overhead equivalent to $2\times$ a masked $GF(2^4)$ multiplier with a single-share input.
\begin{figure}[t!]
	\vspace{-0.3cm}
	\centering
	\includegraphics[width=0.4\textwidth]{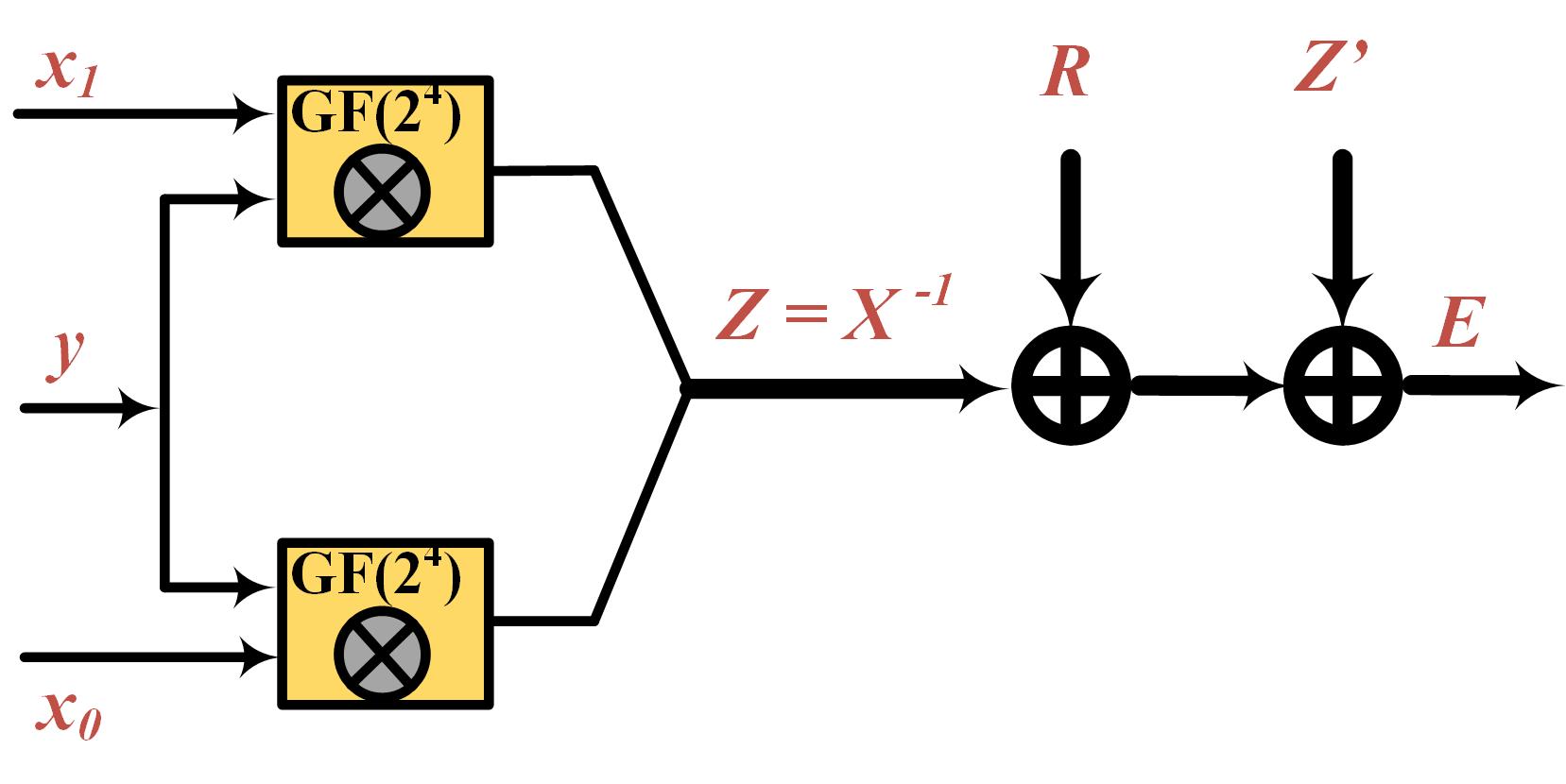}
	\vspace{-0.4cm}
	\caption{Additional redundancy required for protecting RS-Mask against differential fault attacks.}
	\label{fig:redundant}
\end{figure}

The infective RS-Mask scheme with the error-detection as shown in Fig. \ref{fig:redundant} can protect against all statistical fault analysis techniques. However, classical differential fault analysis (DFA) can still attack this scheme. In AES, an error in one byte of the state propagates into the entire column, after \textit{MixColumn}, with deterministic relations of errors in different bytes of the column. The error detection scheme in Fig. \ref{fig:redundant} can also be used to provide security against DFA. Having four random masks $R_i, i=0,\cdots,3$, we can calculate four random variables $E_i=E\times R_i, i=0,\cdots,3$. By adding $E_i$ to the $i$-th byte of a column, the errors of all bytes in a column are randomized with no deterministic relations.

\section{Results} \label{sec:results}
We evaluate the robustness of an FPGA implementation of the RS-Mask against power analysis using the Flexible Open-source workBench fOr Side-channel analysis (FOBOS) \cite{fobos}. Our FOBOS instance uses the NewAE CW305 Artix-7 FPGA target for instantiation of the RS-Mask AES, and Digilent Nexys A7 as the control board for synchronization with a host PC and target FPGA. We employ t-test statistics to show any leakage of sensitive data into the power traces. For fault analysis, we employ a simulated fault injection mechanism into the internal variables of the S-box. We use a similar fault model as described in \cite{dobraunig2018statistical}; the attacker can inject any type of fault into the internal operations of the S-box.
 
\subsection{Resistance against Power Analysis}
In order to demonstrate the resistance of the proposed RS-Mask against power analysis, we use test vector leakage assessment (TVLA) methodology with t-test statistics. TVLA is independent of leakage models while it can detect potential leakages of secret information to the power traces.

To conduct t-test on the RS-Mask scheme, we measured the power traces of FPGA implementation of AES during encryption of multiple plaintexts. Depending on the intermediate variable, i.e., the output of the S-box, the power traces are divided into two separate sets $T_A$ and $T_B$. To detect a first order leakage in the power traces, the t-test statistic is
\begin{equation}
    t = \frac{\mu_A - \mu_B}{\sqrt{\sigma_A^2/N_A + \sigma_B^2/N_B}}
\end{equation}
in which, $\mu_i$, $\sigma_i^2$ and $N_i$ are the sample mean, variance and the number of samples in the power trace set $i\in\{A,B\}$. The results of t-test statistics on 100K measured power traces are shown in Fig. \ref{fig:rs_ttest}. We observe $|t|<4.5$ which implies the lack of first-order leakage with a confidence level of 99.999\%.

\begin{figure}[t!]
	\centering
	\includegraphics[width=0.4\textwidth]{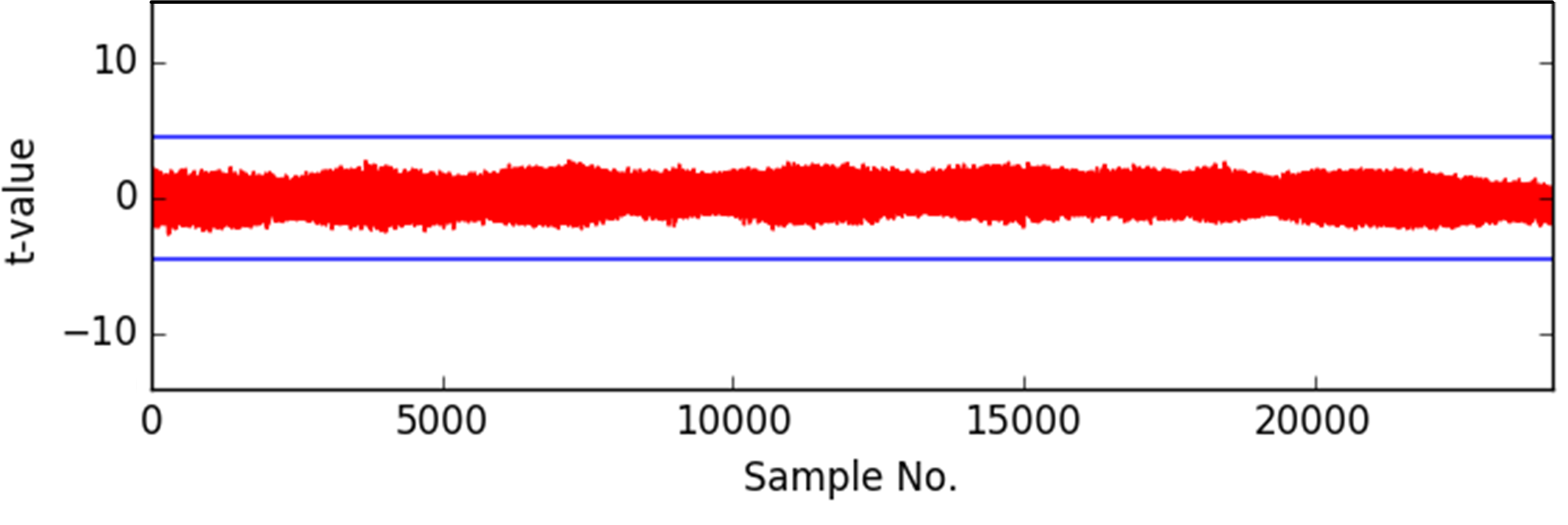}
	\vspace{-0.3cm}
	\caption{T-test statistics on a FPGA implementation of RS-Masked AES with no observable leakage.}
	\vspace{-0.2cm}
	\label{fig:rs_ttest}
\end{figure}

For comparison, the t-test results of an unprotected implementation of AES are shown in Fig. \ref{fig:aes_ttest}. It is observed that the t-values extend beyond the threshold of 4.5 at multiple time samples. It implies that at those time samples, the power consumption of computations are strongly correlated with the processed secret data. Hence, a power analysis technique can likely exploit this correlation to recover the secret key.

\begin{figure}[t!]
	\centering
	\includegraphics[width=0.4\textwidth]{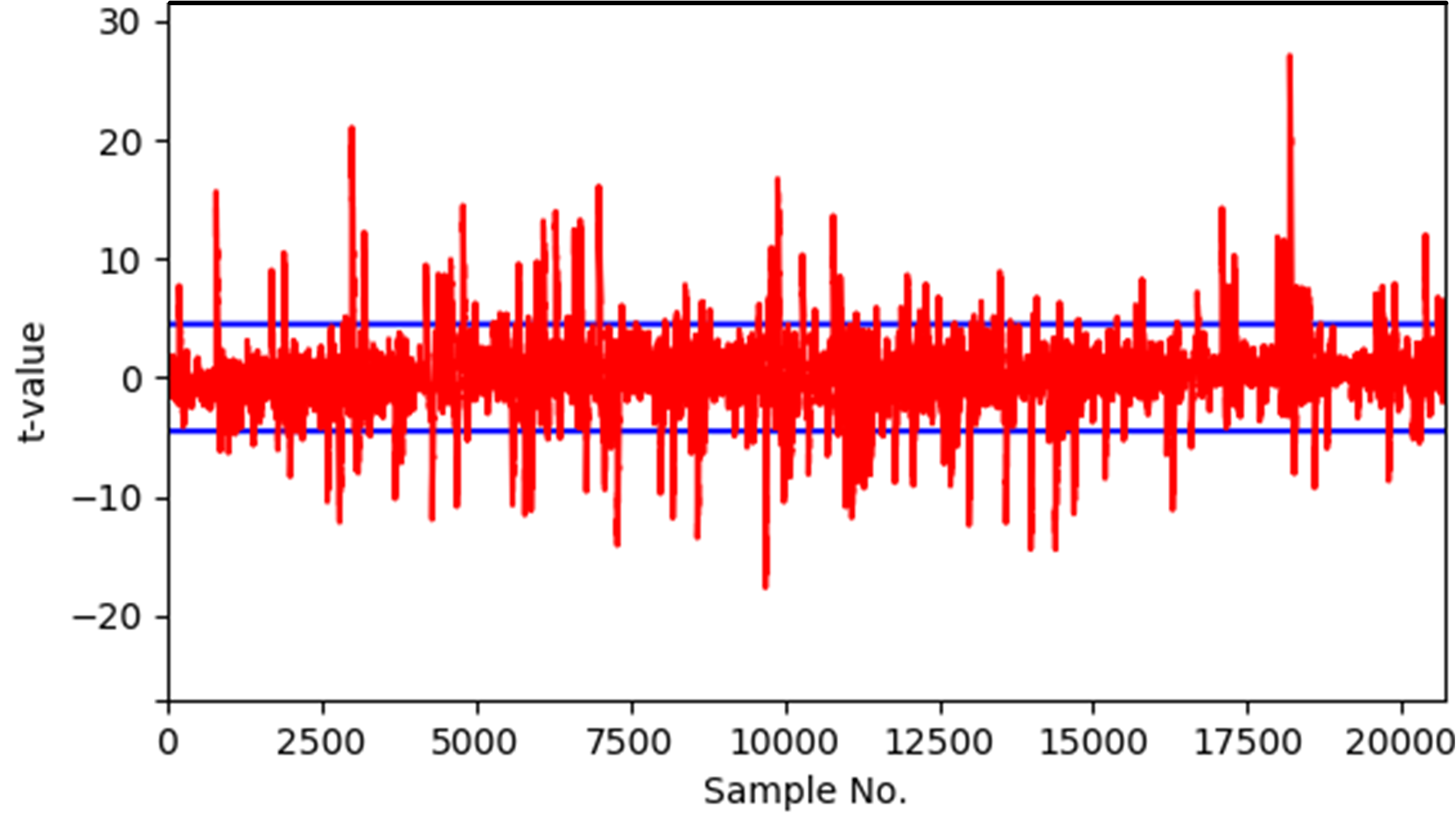}
	\vspace{-0.3cm}
	\caption{T-test statistics on a FPGA implementation of unprotected AES showing leakage of sensitive information.}
	\vspace{-0.3cm}
	\label{fig:aes_ttest}
\end{figure}

\subsection{Resistance against Fault Analysis}
To demonstrate the resistance of the RS-Mask scheme against statistical fault analysis, we inspect the distribution of intermediate variables under fault injection, which ideally must be uniform. Next, we conduct a SIFA attack on an RS-Masked implementation of AES to show that the secret key is not distinguishable using the distribution of intermediate variables.

In Fig. 9, the distribution of faulty values with a single fault injected into the internal computations of the S-box is compared for two masked implementation of AES, i.e., TI-Masked and RS-Masked implementations. The fault location is at the output of the $GF(2^2)$ multiplier, marked with a red circle, shown in Fig. \ref{fig:inverter}.

In Fig. \ref{fig:fault_distr} we observe that, in the TI implementation, a single fault can induce a significant bias into the distribution of values at the output -- even with 3 shares. However, in the RS-Mask scheme, the distribution of faulty values is uniform. For comparison, the distribution of intermediate values calculated with an incorrect key guess is also shown in the figure, which is very similar to the distribution of values calculated with the correct key in the RS-Mask scheme. This result shows that the correct key is indistinguishable in the RS-Mask implementation with any SFA technique that exploits the distribution of faulty values.

\begin{figure}[t!]
	\vspace{-0.3cm}
	\centering
	\includegraphics[width=0.5\textwidth]{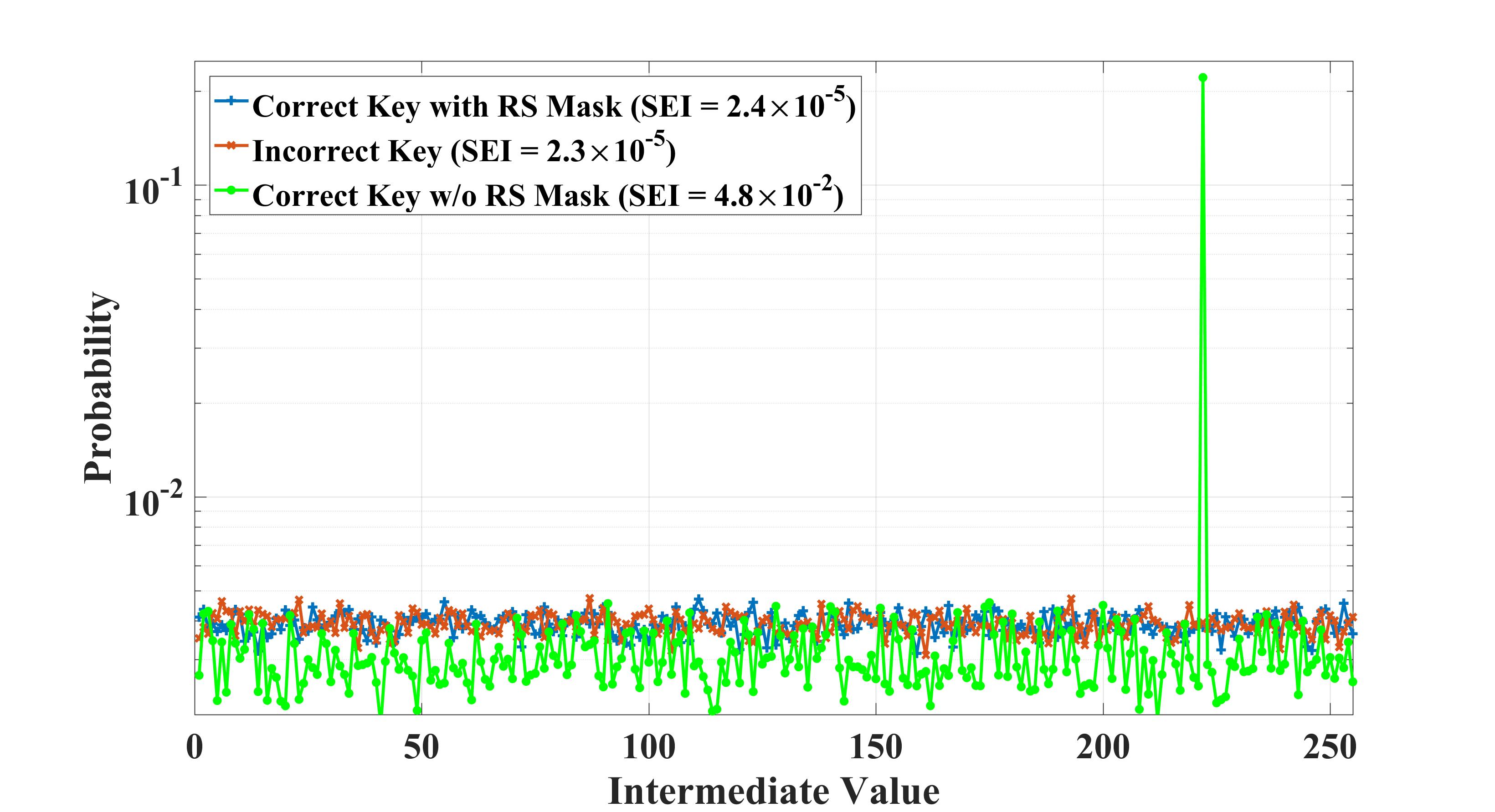}
	\vspace{-0.3cm}
	\caption{Distribution of faulty values at the output of the S-box under attack on TI-Masked and RS-Masked AES implementation.}
	\label{fig:fault_distr}
\end{figure}

The distribution of correct values under ineffective faults is also compared in Fig. \ref{fig:correct_distr}. Similar to the faulty values, the distribution of correct values in the TI implementation is biased as a result of single fault injection. However, in the RS-Mask implementation, distribution of correct values is nearly uniform and is close to the distribution of values calculated with an incorrect key candidate. As a result, the correct key is not distinguishable using correct values.

\begin{figure}[t!]
	\vspace{-0.3cm}
	\centering
	\includegraphics[width=0.5\textwidth]{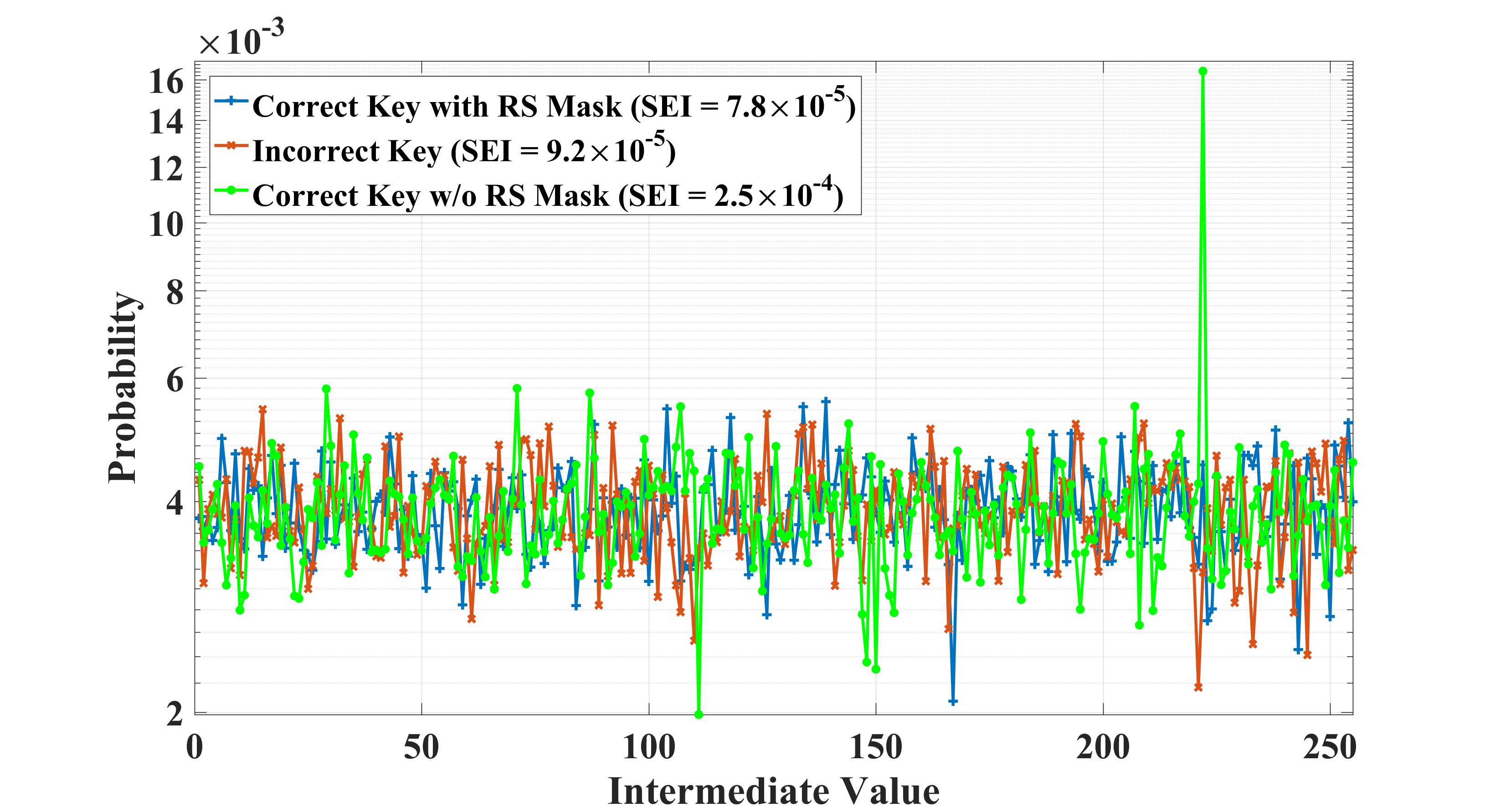}
	\vspace{-0.3cm}
	\caption{Distribution of correct values at the output of the S-box under attack on TI-Masked and RS-Masked AES implementation.}
	\label{fig:correct_distr}
	\vspace{-0.2cm}
\end{figure}

\subsection{Resistance against SIFA attack}
To further demonstrate the resistance of RS-Mask against SIFA, we deploy a fault attack on the RS-Mask and TI implementations of AES with a single fault injected at the internal operations of the S-box at the beginning of round 9. The square Euclidean Imbalance (SEI) of the data distribution with the correct key and the maximum SEI of incorrect key candidates versus the size of data samples is shown in Fig.~\ref{fig:sifa_rs}. We observe that there is always an incorrect key with a larger bias than the correct key, irrespective of the size of data samples used. Hence, the correct key is indistinguishable using the distribution of correct values under fault injection.
\begin{figure}[t!]
	\vspace{-0.2cm}
	\centering
	\includegraphics[width=0.45\textwidth]{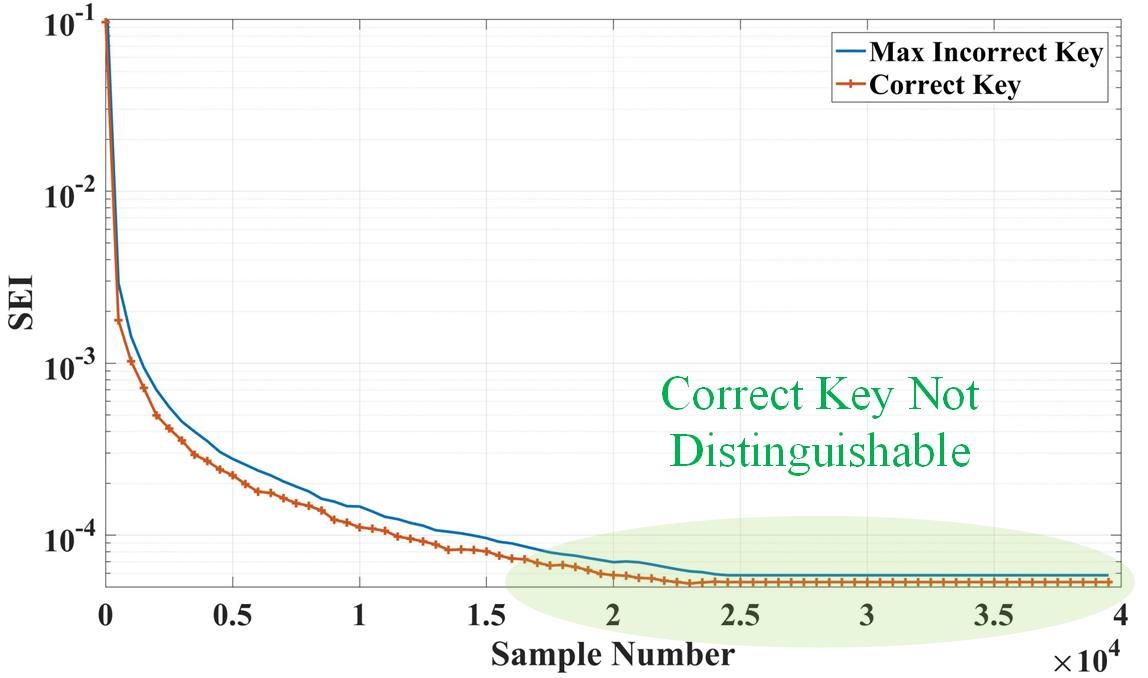}
	\vspace{-0.3cm}
	\caption{Comparing SEI of the correct key with maximum SEI of incorrect key candidates versus the number of correct ciphertexts in RS-Mask.}
	\label{fig:sifa_rs}
\end{figure}

For comparison, the SIFA attack is also deployed on the TI implementation of AES. The SEI of the correct key and the max SEI of incorrect key candidates versus the size of data samples are compared in Fig. \ref{fig:sifa_ti}. After collecting almost 2000 correct ciphertexts, the SEI of correct key is always larger than all incorrect key candidates, which denotes a key recovery.
\begin{figure}[t!]
	\vspace{-0.2cm}
	\centering
	\includegraphics[width=0.45\textwidth]{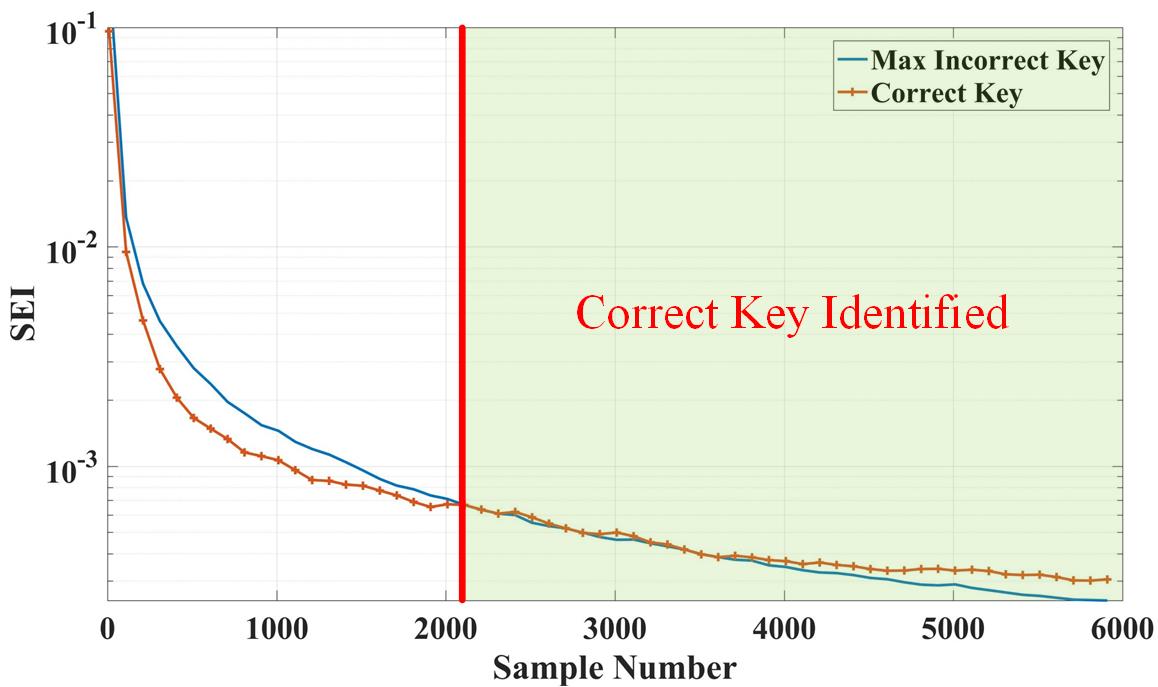}
	\vspace{-0.3cm}
	\caption{Comparing SEI of the correct key with maximum SEI of incorrect key candidates versus the number of correct ciphertexts in TI implementation.}
	\label{fig:sifa_ti}
	\vspace{-0.2cm}
\end{figure}

\section{Conclusion and Future Work}  \label{sec:conclusion}
We proposed random space masking (RS-Mask) as an integrated countermeasure against both power and fault analysis. All sensitive variables in the RS-Mask assume a uniform distribution even under fault injection. We proved that the proposed scheme provides security against statistical fault analysis techniques that observe the distribution of intermediate variables. The security of the scheme against power analysis is based on the guarantees of threshold implementations (TI). We implemented the RS-Mask scheme for the AES algorithm on an FPGA with 10 pipeline stages, and showed that the area of the design with three shares is around $3.5\times$ an unprotected AES and $2\times$ a TI implementation. We employed a SIFA attack on RS-Mask which was unable to detect the correct key. In future work, we will 
further implement infective RS-Mask to demonstrate its robustness to all differential fault analysis techniques with a slight increase in overhead.

\section*{Acknowledgement}

This work was supported by NIST award 70NANB18H219 for Lightweight Cryptography in Hardware and Embedded Systems.

\bibliographystyle{IEEEtran}
\bibliography{IEEEabrv,HOST2020_refs.bib}

\end{document}